\documentclass[cis]{ipart_v1}

\Vol{17}
\Issue{4}
\Year{2017}
\firstpage{193}

\usepackage[english]{babel}

\RequirePackage[OT1]{fontenc}
\RequirePackage{amsthm,amsmath,amssymb}
\RequirePackage{xcolor}

\newtheorem{thm}{Theorem}

\begin{document}

\title[Part I: i.i.d. sources]{Layered black-box, behavioral interconnection perspective and applications to problems in communications, Part I: i.i.d. sources}

\author[M. Agarwal, S. Mitter, and A. Sahai]{Mukul Agarwal$^\dag$, Sanjoy Mitter$^\ddag$, and Anant Sahai\blfootnote{$^\dag$This work is a part of the authors Doctoral research, carried out in the Laboratory for Information and Decision Systems, Department of Electrical Engineering and Computer Science, Massachusetts Institute of Technology, Cambridge MA 02139.}
\blfootnote{$^\ddag$This research has been supported by NSF Grant CCF-0325774 ITR :: Collaborative Research : New Approaches to Experimental Design and Statistical Analysis and Structural Biologic Data from Multiple Sources (2003-2009 )   and NSF grant EECS-1135843  CPS : Medium : Collaborative Research : Smart Power Systems of the Future : Foundations for Understanding Volatility and Improving Operational Reliability ( 2011--2017 ).}}

\begin{abstract}
In this paper, the problem of communication over an essentially unknown channel, which is known to be able to communicate a source to a destination to within a certain distortion level, is considered from a behavioral, interconnection view-point. Rates of reliable communication are derived and source-channel separation for communication with fidelity criteria is proved. The results are then generalized to the unicast multi-user setting under certain assumptions. Other applications of this problem  problem which follow from this perspective are discussed.
\end{abstract}

\maketitle

\section{Introduction}

In this paper, the problem of communication over an essentially unknown channel which is known to be able to communicate a source to a destination to within a certain distortion level, is considered from a behavioral, interconnection view-point.

Let $\mathbb I$ and $\mathbb O$ be two finite sets. Consider a channel $k$ with input space $\mathbb I$ and output space $\mathbb O$. Let $\mathbb X$ and $\mathbb Y$ be finite sets. Let $e$ be an encoder with input space $\mathbb X$ and output space $\mathbb I$ and $f$ be a decoder with input space $\mathbb O$ and output space $\mathbb Y$. In the information theoretic setting, the encoder, channel and decoder are all sequences $e = <e^n>_1^\infty$, $k = <k^n>_1^\infty$ and $f = <f^n>_1^\infty$ where superscript $n$ denotes block-length. $k$ may be a memoryless or a Markoff channel, or more generally, a general channel in the sense of Verdu-Han \cite{VerduHan}, that is, $k^n: \mathbb I^n \rightarrow \mathbb O^n$ and $k^n(o^n|\iota^n)$ is the probability that the channel output is $o^n$ given that the input is $\iota^n$. An important point is that $c = <c^n>_1^\infty = <e^n \circ k^n \circ f^n>_1^\infty$ is a general channel in the sense of Verdu-Han with input space $\mathbb X$ and output space $\mathbb Y$. For this reason, if statements can be made about a general channel $c$, and if $c$ is or can be thought of as the composition of an encoder $e$, channel $k$ and decoder $f$, a similar statement can made about the channel $k$. Note that even if $k$ has special structure (memoryless/ Markoff/ indecomposable), $c$ is a general channel. This paragraph is elaborated in greater detail below.

Let $c$ be a general channel with input space $\mathbb X$ and output space $\mathbb Y$. Let a guarantee $G$ (for example, reliable communication or communication within a certain distortion or communication within a certain error probability or something else) be met over $c$ for the \emph{direct} transmission or a random source $X = <X^n>_1^\infty$ on $\mathbb X$. That is, with input $X$ to $c$, the output is $Y = <Y^n>_1^\infty$ such that guarantee $G$ is met between $X$ and $Y$.  Let it be possible to communicate a random source $X' = <{X'}^n>_1^\infty$ on $\mathbb X'$ over $c$ with some guarantee $G'$. This means the following. There exists an encoder $e'= <{e'}^n>_1^\infty$ with input space $\mathbb X'$ and output space $\mathbb X$ and a decoder $f' = <{f'}^n>_1^\infty$ with input space $\mathbb Y$ and output space $\mathbb Y'$ such that with input $X'$ to $e'$, the output of $f'$ is $Y' = <{Y'}^n>_1^\infty$ such that guarantee $G'$ is satisfied between $X'$ and $Y'$. Then, if c were the composition of encoder, channel, decoder $e, k, f$, that is, $<c^n>_1^\infty = <e^n \circ k^n \circ f^n>_1^\infty$ it follows that by use of encoder $e' \circ e = <e'^n \circ e^n>_1^\infty$ and decoder $f \circ f' = <f^n \circ f'^n>_1^\infty$ guarantee $G'$ for communication of random source $X'$ over $k$ is met. Note that it is crucial that the model for $c$ is a general channel model. This is because the composition of encoder, channel and decoder is a general channel in the sense of \cite{VerduHan} and not a `physical' channel even if $k$ is a physical channel.

Given finite sets $\mathbb X$ and $\mathbb Y$. An architecture with input space $\mathbb X$ and output space $\mathbb Y$ is a transition probability $a: \mathbb X \rightarrow \mathbb P(\mathbb Y)$. Given an architecture $a$. If $b = a \circ d$ or $b = d \circ a$, $d$ is said to be layered `above' $a$ in order to get the architecture $b$.

With this terminology, $<e', f'>$ is said to be layered `above' $<e \circ k \circ f>$. This methodology is said to be a layering methodology.


The advantage of the layering methodology is inherent in the above statement:
\begin{itemize}
\item
If a statement can be made about $e \circ k \circ f$ viewed as a general channel, then, a similar statement can be made about $k$, and further, the model of $k$ may be general.
\end{itemize}

If in addition, in order to prove that guarantee $G'$ is met over $c$, only the knowledge of the fact that source $X$ is communicated \emph{directly} with guarantee $G$ over $c$ is used, and not the internal functioning of $c$. This perspective is called a \emph{black-box perspective}. Thus, $c$, or $e \circ k \circ f$ is seen as a black box which accomplishes a certain end-to-end guarantee $G$, and then, based on the knowledge of this guarantee, and not the internal functioning of $e$, $k$ or $f$, it is proved that guarantee $G'$ is met over $k$. 

The black-box perspective has the following advantage: 
\begin{itemize}
\item
Results proved by use of the black-box perspective are, in general, \emph{universal over the channel}, that is, the channel can belong to a set. This is due to the fact that only the knowledge that source $X$ is communicated \emph{directly} with guarantee $G$ over $c$ is used. Universal results are important because channels for communication, for example, wireless channels or the internet are, in general, unknown or only partially known.
\item
Roughly speaking, the channel model does not matter. All that matters is the end-to-end description of the channel.
\end{itemize}

Finally, let end-to-end guarantee $G$ for \emph{direct} communication of source $X$ over $c$ be met. If, in order to prove that guarantee $G'$ for communication of $X'$ is met over $c$, the encoder-decoder $(e, d)$ are so constructed that $X'' \triangleq e(X')$ has the same distribution as $X$. This perspective is called a \emph{behavioral or a behavioral interconnection} perspective. The use of the phrase `behavioral perspective' is borrowed from behavioral systems theory as developed by Jan Willems \cite{Willems}, \cite{WillemsCSM}. If $Y$ is the output of $c$ when $X$ is the input, $(X,Y)$ is the behavior of the channel $c$. From the above definition, $(X', X'')$ is a behavior of $e$ where $X''$ has the same distribution as $X$. Then, by interconnecting $e$ and $c$, $(X', X, Y)$ is a behavior of $(e, c)$ Thus, \emph{behaviors are matched at the point of interconnection}. The behavioral view-point is further elaborated on, in \cite{AgarwalMItter}.

There are various advantages of the interconnection perspective. 
\begin{itemize}
\item
It is intuitively appealing: if the black-box perspective is used, the common-sense way to prove a theorem concerning communication of $X'$ over $c$ is to simulate the source $X$ at the input of the channel; in other words, behaviors are matched, and this is the behavioral, interconnection perspective. 
\item
Statements can be made regarding channel resource consumption.\linebreak  With input $X$ to $c$ the output is $Y$. The channel resource consumption, for example, energy or bandwidth consumption  is a function of the distribution $p_X$ of $X$ (this point is discussed in detail later). Now, in order to communicate $X'$ over $c$, $X'' = e(X')$ has the same distribution as $X$. Thus, channel resource consumption is the same for communication of $X'$ over $c$ as it is for the  communication of $X$ over $c$.
\end{itemize}

\emph{The layered black-box behavioral interconnection view-point will be illustrated in this paper in the following way:}
\begin{itemize}
\item
Results concerning communication with fidelity criteria over a medium will be proved
\item
An operational view of source-channel separation for communication with a fidelity criterion and a randomized covering-packing duality between source and channel coding which have been discussed in \cite{paperDuality} will be re-visited from the view-point proposed in this paper
\end{itemize}

To elaborate:
\begin{itemize}
\item
Consider a set of channels $\mathbb A$. A generic channel in this set is denoted by $c = <c^n>_1^\infty$ and has input space $\mathbb X$ and output space $\mathbb Y$; $\mathbb X$ and $\mathbb Y$ are finite sets. Such a channel is called a general, compound channel. Let it be known that $c \in \mathbb A$ communicates the i.i.d. $X$ source within a distortion $D$. Mathematically, this means the following: 
Let $X^n$ denote the i.i.d. $X$ sequence of length $n$. With input $X^n$ to $c^n$, the output is $Y^n$ such that $\exists <\omega_n>_1^\infty$, $\omega_n \in \mathbb R, \omega_n \to 0$ as $n \to \infty$ such that
\begin{align}
\Pr \left ( \frac{1}{n} d^n(X^n, Y^n) > D \right ) \leq \omega_n \quad \forall c \in \mathbb A
\end{align}
Note that $Y^n$ depends on the particular $c \in \mathbb A$ but this dependence is suppressed. It will be proved by use of a black-box, behavioral interconnection view-point  that if random codes are  permitted, the universal or the compound channel capacity of the set of channels $\mathbb A$ is  $\geq R_X(D)$ by use of the same channel resources as are used in the communication of the i.i.d. $X$ source over $c \in \mathbb A$. $R_X(D)$ is the rate-distortion function of the i.i.d. $X$. The compound capacity refers  to the capacity where the same encoder-decoder should work for every $c \in \mathbb A$ for reliable communication.
\item
  Further, by use of the above result and by the discussion on the layered black-box interconnection perspective carried out before, a universal source-channel  separation theorem for communication with a fidelity criterion will be proved where universality is over the channel. This requires the use of random codes. This is non-trivial because if random codes are not permitted, universal source channel separation for communication with a fidelity  criterion may not hold.
\item
  A generalization to the multi-user setting of the above mentioned results will be presented, again by use of the layered black-box behavioral view-point. As will be seen, the behavioral view-point is crucial here, and realizing that makes the multi-user problem an essentially trivial generalization of the point-to-point problem. In the multi-user setting, it will be assumed that independent sources need to be communicated between various source-destination pairs; more precisely, if source $X_{ij}$ needs to be communicated from user $i$ to user $j$, sources $X_{ij}$ are assumed to be independent. Note that this is the unicast setting; the same source $X_{ij}$ does not need to be reproduced at various destinations.
\item On a conceptual front,
  \begin{itemize}
\item
An operational view-point on source-channel separation theorem for communication with a fidelity criterion which uses only the operational meanings of channel capacity as the maximum rate of reliable communication and the rate-distortion  function as the minimum rate needed to code a source with a certain distortion level are used has been discussed for the uniform $X$ source in \cite{paperDuality}. The same will be discussed in brief in this paper but  the layered black-box behavioral nature of the view-point will be stressed. This is unlike  \cite{Shannon} which uses information theoretic quantities, namely, channel  capacity defined as a maximum mutual information and rate-distortion function defined as a minimum mutual information crucially. The operational view-point has the following advantage: 
\begin{itemize}
\item
A channel can be used for various purposes. For example, it can be used for reliable communication, it can be used for transmission of a control signal, and so on. An information-theoretic characterization of a channel which is used for reliable communication is done in terms of a maximum mutual information. However, if the same channel is used for control purposes, this characterization may not be justified. For this reason, view-points which use only the operational meanings of quantities, for example, channel capacity defined as the maximum rate of reliable communication over a channel for the problem of communication or the channel efficiency defined as the maximum noise variance acceptable for stabilization over a noisy communication link are desirable instead of special mathematical characterizations.
\end{itemize}
\item
A randomized covering-packing duality between source-coding and channel-coding has been discussed in \cite{paperDuality}. The same will be discussed in brief in this paper but the layered black-box behavioral nature of the view-point will be emphasized.
\item
  Another application is that the converse part of source-channel separation theorem for communication with a fidelity criterion which is proved by use of converse techniques in  for example, \cite{Shannon} is proved using achievability techniques in this paper.
\item
  An equivalence is demonstrated between the problem of reliable communication and communication within a certain distortion in the point-to-point setting.
\end{itemize}
\end{itemize}

Some of the results for point-to-point communication in this paper appeared in \cite{AgarwalSahaiMitterAllerton}. It was proved by use of a layered black-box reduction view-point that if a general compound channel communicates the i.i.d. $X$ source to a destination within a certain distortion level, reliable communication can be accomplished over the channel at rates less than the rate-distortion function corresponding to this distortion level.

{\color{red}} An individual sequences approach for this problem has been taken in \cite{LomnitzFeder}. A communication problem is considered where the model of the channel is left unspecified. The achievable rate is determined as a function of the channel input and output sequences known aposteriori. It is proved that rates less than the empirical mutual information between the input and output are achievable. Our approach generalizes to certain multi-user settings This is discussed below.

A layering black-box approach, called the concatenated approach has been used in \cite{ForneyD}. Suppose a coder and decoder is set up for some channel for block-length $n$, then the composition of coder, channel, decoder can be thought of as a larger channel, and $r$ copies of this can be thought of as a super-channel. Now, the super-channel can be thought of as a single channel. The composition of the inner code with the channel can be viewed as a black box with a certain performance guarantee, and an outer code around it can be designed by use of a layering methodology. See \cite{ForneyD} for details.

A similar view has also been used in \cite{Tian} for the problem of proving the optimality of separation for communication with fidelity criteria in a network setting. However, they require a memorylessness assumption (or more generally, a finite-state assumption) on the medium of communication and their results are not universal over the medium of communication. A more thorough comparison of \cite{Tian} with our work will be carried out later in the paper.

A reduction view has been used, for example, in  \cite{Effros1}, \cite{Effros2} where a family of equivalence tools for bounding network capacities is introduced. The idea is not to compute capacity regions but derive relations between capacity regions especially when these capacity regions remain inaccessible. In our paper, a similar reduction strategy is followed: the problem of communicating bits reliably over a channel is reduced to the problem of communication with a fidelity criterion over a channel where the only knowledge of the channel that is used is that it communicates the source within a certain distortion level (thus, for example, no knowledge of channel capacity or channel transition probability is used), and this leads to a source-channel separation theorem for communication with a fidelity criterion which can be generalized to multi-user settings.

\section{Outline of the paper}

In Section \ref{NotandDef}, basic notation used throughout the paper is presented. This is followed by the proof of the basic theorem on which all theorems in this paper are built in Section \ref{BBTT}. Section \ref{UUTT} uses the theorem from Section \ref{BBTT} in order to prove a universal source-channel separation theorem for communication with a fidelity criterion. Results of Sections \ref{BBTT} and \ref{UUTT} are generalized to the multi-user setting in Section \ref{MMTT}. Section \ref{KKKKKKK} discusses some applications of  conceptual nature. An operational view-point on source-channel separation for communication with a fidelity criterion and a randomized covering-packing duality between source and channel coding which have been discussed in \cite{paperDuality} in detail is discussed in very brief in Subsection \ref{OOVV}. Subsection \ref{AATT} provides an alternate proof of the converse part of source-channel separation for communication with a fidelity criterion by use of achievability techniques. Subsection \ref{Equ} provides an equivalence between the problem of reliable communication and that of communication within a certain distortion level. Section \ref{RR} recapitulates this paper.

\section{Basic notation and definitions}\label{NotandDef}. 

Basic notation is defined in this section. Further notation will be defined as the paper proceeds.

Superscript $n$ will denote a quantity related to block-length $n$. For example, $x^n$ will be the channel input when the block-length is $n$. The only exception to this rule is for a real number: $\omega_n$ is used corresponding to block-length $n$ if $\omega_n$ is a real number in order not to confuse $\omega^n$ with the $n^{th}$ power of $\omega$. As the block-length varies, $x = <x^n>_1^\infty$ will denote the sequence for various block-lengths.

Let $\mathbb X$ and $\mathbb Y$ be finite sets. $\mathbb X$ is the source input space. $\mathbb Y$ is the source reproduction space. $\mathbb X$ is also the channel input space and $\mathbb Y$ is also the channel output space.

Let $x^n \in \mathbb X^n$. $p_{x^n}$ will denote the empirical distribution or the type of $x^n$, that is, for $ x \in \mathbb X$,
\begin{align}
p_{x^n} \triangleq \frac{\mbox{number of $x$ in the sequence $x^n$}}{n}
\end{align}

Let $p$ be a probability distribution on $\mathbb X$.  Let $\epsilon \geq 0$. The sequence $x^n$ is said to belong to $\mathbb T(p, \epsilon)$ if
\begin{align}
\sum_{x \in \mathbb X} |p_{x^n}(x) - p(x)| \leq \epsilon
\end{align}

Let $X$ be a random variable on $\mathbb X$.  $p_X$ will denote the probability distribution corresponding to $X$. 

Let $X^n$ denote the i.i.d. $X$ sequence of length $n$. The i.i.d. $X$ source is $<X^n>_1^\infty$. 

$d: \mathbb X \times \mathbb Y \rightarrow [0, \infty)$ is a distortion function. The $n$-letter distortion function is defined additively:
\begin{align}
d^n(x^n, y^n) = \sum_{i=1}^n d(x^n(i), y^n(i)),\quad x^n \in \mathbb X^n, y^n \in \mathbb Y^n
\end{align}
where $x^n(i)$ denotes the $i^{th}$ coordinate of $x^n$ and similarly for $y^n$.

The channel is a sequence $c = <c^n>_1^\infty$ where
\begin{align}
c^n &: \mathbb X^n \rightarrow \mathbb P(\mathbb Y^n) \\
       &: x^n \rightarrow c^n(\cdot|x^n)
\end{align}
When the block-length is $n$, the channel acts as $c^n(\cdot|\cdot)$; $c^n(y^n|x^n)$ is the probability that the channel output is $y^n$ given that the channel input is $x^n$. This model is the same as the model in \cite{VerduHan}. The model of a ``physical'' channel would impose causality and nestedness among various  $c^n$ and this would be a special case of the above model. A compound channel is a set $c \in \mathbb A$ of channels with input space $\mathbb X$ and output space $\mathbb Y$. This model of a compound channel is the same as the model of a compound channel in \cite{CsiszarKorner} though the emphasis in \cite{CsiszarKorner} is on compound DMCs. This model of a compound channel also generalizes arbitrarily varying channels defined in \cite{CsiszarKorner}.

Next, source codes and channel codes are defined.

Let 
\begin{align}
\mathbb M^{n}_R \triangleq \{1, 2, \ldots, 2^{\lfloor nR \rfloor} \}
\end{align}
$\mathbb M^{n}_R$ is the message set. When the block-length is $n$, a rate $R$ deterministic source encoder  is $e_s^{n}: \mathbb X^{n} \rightarrow \mathbb M_R^{n}$ and a rate $R$ deterministic source decoder $f_s^{n}: \mathbb M_R^{n} \rightarrow \mathbb Y^{n}$.  $(e_s^{n}, f_s^{n})$ is the block-length $n$ rate $R$  deterministic source-code. The source-code is allowed to be random in the sense that encoder-decoder is a joint probability distribution on the space of deterministic encoders and decoders. $<e_s^{n}, f_s^{n}>_1^\infty$ is the rate $R$ source-code. The classic argument used in \cite{Shannon} to prove the achievability part of the rate-distortion theorem uses a random source code.

When the block-length is $n$, a rate $R$ deterministic channel encoder is a map $e_{ch}^{n}:\mathbb M_R^{n} \rightarrow \mathbb X^{n}$ and a rate $R$ deterministic channel decoder is a map $f_{ch}^{n}: \mathbb Y^{n} \rightarrow \hat {\mathbb M}_R^{n}$ where $\hat {\mathbb M}_R^{n} \triangleq \mathbb M_R^{n} \cup \{e\}$ is the message reproduction set where `e' denotes error. The encoder and decoder are allowed to be random in the sense discussed previously. $<e_{ch}^{n}, f_{ch}^{n}>_1^\infty$ is the rate $R$ channel code. The classic argument used in \cite{ShannonReliable} to derive the achievability of the mutual information expression for channel capacity uses a random channel code.

The source-code $<e_s^{n'}, f_s^{n'}>_1^\infty$ is said to code the i.i.d. $X$ source within a probability of excess distortion $D$ if with input $X^{n}$ to $e_s^{n} \circ f_s^{n}$, the output is $Y^{n}$ such that
\begin{align} \label{SourceCodeD}
\lim_{n \to \infty} \Pr \left ( \frac{1}{n} d^{n}(X^{n}, Y^{n}) > D \right ) = 0
\end{align}
(\ref{SourceCodeD}) is the probability of excess distortion criterion.
The infimum of rates needed to code the source $X$ source within the distortion $D$ under the probability of excess distortion criterion is the rate-distortion function $R^P_X(D)$.

The source-code $<e_s^{n'}, f_s^{n'}>_1^\infty$ is said to code the i.i.d. $X$ within an expected distortion $D$ if with input $X^{n}$ to $e_s^{n} \circ f_s^{n}$, the output is $Y^{n}$ such that
\begin{align} \label{SourceCodeED}
\lim_{n \to \infty} E \left [ \frac{1}{n} d^n(X^n, Y^n \right ] \leq D
\end{align}

(\ref{SourceCodeD}) is the expected distortion criterion. The infimum of rates needed to code the i.i.d. $X$ source within distortion $D$ under the expected distortion criterion is the rate-distortion function $R^E_X(D)$.

The information theoretic rate-distortion function is defined as 
\begin{align}\label{JInformationTheoreticRateDistortionFunction}
R^I_X(D) \triangleq \inf_{ \left \{ p_{Y|X}\ : \   \sum_{x \in \mathbb X, y \in \mathbb Y} p_X(x)p_{Y|X}(y|x) \leq D \right \} } I(X;Y) 
\end{align}
It is known that $R^P_X(D), R^E_X(D)$ and $R^I_X(D)$ are equal.

Next, reliable communication is considered.

Denote 
\begin{align}
g = <g^{n}>_1^\infty  \triangleq <e_{ch}^{n} \circ c^{n} \circ f_{ch}^{n}>_1^\infty
\end{align}
$g^n$ has input space $\mathbb M_R^{n}$ and output space $\hat {\mathbb M}_R^{n}$. Consider the set of channels
\begin{align}
\mathbb G_{\mathbb A} \triangleq \{ e \circ c \circ f \ | \ c \in \mathbb A \} 
\end{align}
$g \in \mathbb G_{\mathbb A}$ is a compound channel.  Rate $R$ is said to be reliably achievable over $g \in \mathbb G_{\mathbb A}$ if there exists a  rate $R$ channel code $<e_{ch}^{n}, f_{ch}^{n}>_1^\infty$  and a sequence $<\delta_n>_1^\infty$, $\delta_n \to 0$ as $n \to \infty$ such that 
\begin{align}
\sup_{m^{n} \in \mathbb M_R^{n}} g^{n}(\{ m^{n}\}^c|m^{n}) \leq \delta_n \quad \forall c \in \mathbb A
\end{align}
Supremum of all achievable rates is the capacity of $c \in \mathbb A$. \emph{Note that this is the compound capacity, but will be referred to as just the capacity of $c \in \mathbb A$.}

Next, a channel which communicates a source directly within a certain distortion is defined. 

The channel $c \in \mathbb A$ is said to communicate the source $X$ \emph{directly} within distortion $D$ if with input $X^{n}$ to $c^{n}$, the output is $Y^{n}$ (possibly depending on the particular $c \in \mathbb A$) if such that $\exists$ $\omega_n \to 0$ as $n \to \infty$ such that 
\begin{align} 
\Pr \left ( \frac{1}{n} d^{n}(X^{n}, Y^{n}) > D \right )  \leq \omega_n\quad  \forall c \in \mathbb A
\end{align}

In this paper, only the end-to-end description of a channel $<c^{n}>_1^\infty$ which communicates  the i.i.d. $X$ source directly within distortion $D$ is used and not the particular $c^{n}$; for this  reason, the {  general channel} should be thought of as a \emph{black-box} which communicates  the i.i.d. $X$ source within distortion $D$.

Next, channel resource consumption, for example, bandwidth consumption, is considered.

In the above, the input to the channel $c$ is i.i.d. $X$. The bandwidth consumption is a function only of $p_X$. In general, it is assumed that any kind of channel resource consumption depends only on $p_X$. Thus, if the channel is used for another purpose and the input to the channel is still i.i.d. $X$, the new architecture will be said to have the same channel resource consumption as the original architecture.

This finishes the section on basic notation and definitions.

\section{Basic theorem} \label{BBTT}

In this section, the basic theorem on which other theorems are built is proved and the proof is followed by comments.
\subsection{Basic theorem and its proof} \label{KKiidiidKK}
\begin{thm}\label{BasicTheorem}
Let $c \in \mathbb A$ be a set of channels which directly communicate the i.i.d. $X$ source within distortion $D$. Then, assuming random codes are permitted, the capacity of this set of channels is $\geq R^I_X(D)$. This reliable communication can be accomplished by use of the same channel resources as consumed in the communication of the i.i.d. $X$ source over $c \in \mathbb A$.
\end{thm}
\begin{proof}
A random-coding technique will be used in order to prove the theorem.

\emph{Codebook generation:} Generate $2^{\lfloor nR \rfloor}$ codewords i.i.d. $X$. This is the codebook $\mathbb K^n$.

\emph{Joint typicality:} $(x^n, y^n)$, $x^n \in \mathbb X^n$, $y^n \in \mathbb Y^n$, are said to be $\epsilon$ jointly typical if 
\begin{itemize}
\item
$x^n$ is $p_X$ typical, that is, $x^n \in \mathbb T(p_X, \epsilon)$
\item
$\frac{1}{n}d^n(x^n, y^n) \leq D$
\end{itemize}

\emph{Decoding:} Let $y^n$ be received at the output of the channel. If $\exists$ unique $x^n \in $ the codebook $\mathbb K^n$ such that $(x^n, y^n)$ is $epsilon$ jointly typical, declare that $x^n$ is transmitted, else declare error.

For block-length $n$, the encoder-decoder pair will be denoted by $( e_{ch}^n, f_{ch}^n)$.

Denote: $<g_c^n>_1^\infty = < e_{ch}^n \circ c^n \circ f_{ch}^n>_1^\infty$.

Denote: 
\begin{itemize}
\item $x^n \in \mathbb K^n$ is the transmitted codeword corresponding to a a particular message $m^n$
\item $y^n \in \mathbb Y^n$ is the received sequence 
\item $z^n \in \mathbb K^n$ is a codeword which is \emph{not} transmitted
\end{itemize}

$g_c^n(\{m^n\}^c|m^n)$ needs to be calculated.

\begin{align}
g_c^n(\{m^n\}^c|m^n) \leq  \Pr(\mathbb E_1^n  \cup \mathbb E_2^n)
                                                                     \leq \Pr(\mathbb E_1^n) + \Pr(\mathbb E_2^n)
\end{align}
where 
\begin{itemize}
\item
$\mathbb E_1^n$:  $(x^n, y^n)$ is not $\epsilon$ jointly typical
\item
$\mathbb E_2^n$: $\exists z^n \neq x^n$ such that $(z^n, y^n)$ is $\epsilon$ jointly typical
\end{itemize}

By definition of $c \in \mathbb A$, $\Pr(\mathbb E_1^n) \to 0$ as $n \to \infty$ at a uniform rate over $c \in \mathbb A$, and independently of $m^n$.

Calculation of $\Pr(\mathbb E_2^n)$ requires a method of types calculation, see \cite{CsiszarKorner} or \cite{CoverThomas}, similar in arguments to those used in the proof of Lemma 1 in \cite{LomnitzFeder} which, in turn, are similar to the error analysis in the conference version \cite{AgarwalSahaiMitterAllerton}. This calculation is carried out below:

Fix $y^n$. Let $y^n$ have type $q_Y$.  Let $z^n$ be a particular non-transmitted codeword. Recall that $z^n$ is generated i.i.d. $X$ and is independent of $y^n$.  Let $q_{Z|Y}$ be a conditional distribution $\mathbb Y \rightarrow \mathbb P(\mathbb X)$. $q_{ZY}$ is the joint distribution resulting from $q_Y$ and $q_{Z|Y}$. Consider the set
\begin{align}
\{z^n \in \mathbb X^n\ | \ p_{z^n|y^n} = q_{Z|Y} \}
\end{align}
By Sanov's Theorem \cite{CoverThomas}, it follows that the probability of this set under the distribution $p_{X^n}$ is 
\begin{align}
\leq & \prod_{y \in \mathbb{Y}} 2^{-nq_Y(y)D(q_{Z|Y}(\cdot|y)||p_X)} \nonumber \\
=     & 2^{-n \sum_{y \in \mathbb{Y}}q_Y(y)D(q_{Z|Y}(\cdot|y)||p_X)} \nonumber \\
=     & 2^{-n D(q_{ZY}||p_X q_Y)}
\end{align}
It follows by
\begin{itemize}
\item
noting that $q_Y$ is arbitrary,
\item
number of joint types on $\mathbb X \times \mathbb Y$ is $\leq (n+1)^{|\mathbb X||\mathbb Y|}$,
\item
and by use of the union bound,
\end{itemize}
that 
\begin{align}\label{FinalE2Bound}
  \Pr(\mathbb E_2^n) \leq (n+1)^{|\mathbb{X}||\mathbb{Y}|} 2^{\lfloor nR \rfloor}  2^{-n \inf _ {q_{ZY}\in \mathbb{T}} D(q_{ZY}||p_X q_Y) }
\end{align}
where
\begin{align}
\mathbb{T} \triangleq  \left \{ q_{ZY}: 
       \begin{array}{l} q_Z \in p_X \pm \epsilon, \mbox{that is,} \sum_{x \in \mathbb X} |p_X(x) - q_Z(x)| \leq \epsilon \\
         \sum_{x \in \mathbb X, y \in \mathbb Y} q_{ZY}(x, y)d(x, y) \leq D \\
       \end{array} \right \}
\end{align}
The above bound on $\Pr(\mathbb E_2^n)$ is independent of $m^n$ and $c \in \mathbb A$. As stated before, $\Pr(\mathbb E_1^n) \to 0$ as $n \to \infty$ at a uniform rate over $c \in \mathbb A$, and independently of $m^n$. 

Thus, rates $R < \inf_ {q_{ZY} \in \mathbb T} D(q_{ZY}||p_X q_Y)$ are reliably achievable over $c \in \mathbb A$. Now,
\begin{align}
& D(q_{ZY}||p_X q_Y) =  
D(q_{Z}||p_X) + D(q_{ZY}||q_Z q_Y)  \geq    D(q_{ZY}||q_Z q_Y)
\end{align}
Thus, rates 
\begin{align}
R < \inf_ {q_{ZY} \in \mathbb T} D(q_{ZY}||q_Z q_Y) =  \inf_ {q_{ZY} \in \mathbb T} I(Z;Y)
\end{align}
are reliably achievable over $c \in \mathbb A$. Now,
\begin{align}
 \inf_{q_{ZY} \in \mathbb T} I(Z;Y) =  \inf_{Z \in \mathbb T(p_X, \epsilon)} R^I_Z(D)
\end{align}
Thus, rates 
\begin{align}
R < \inf_{Z \in \mathbb T(p_X, \epsilon)} R^I_Z(D) 
\end{align}
are reliably achievable over $c \in \mathbb A$. Now, $\epsilon > 0$ is arbitrary and $R^I_X(\cdot)$ is continuous. Thus, rates $R <  R^I_X(D)$ are reliably achievable over $c \in \mathbb A$. 

{
Note that the code  uses random codes generated with distribution $p_X$. It follows from the discussion in Section \ref{NotandDef} that resource consumption for the architecture for reliable communication consumes the same channel resources as needed for communication of the i.i.d. $X$ source.}

This finishes the proof of Theorem \ref{BasicTheorem}.
\end{proof}

\subsection{Comments}

Note that the argument uses a black-box, interconnection view-point:
\begin{itemize}
\item
The decoding rule uses only the input-output characterization of the channel that i.i.d. $X$ source is communicated over the channel with distortion $D$ and not the particular $c^n$ of the channel $c = <c^n>_1^\infty$. Thus, the perspective is a black-box perspective
\item
The encoder output is i.i.d. $X$ and has the same distribution as the input to the channel which communicates the i.i.d. $X$ source within distortion $D$. Thus, the perspective is a behavioral interconnection perspective
\end{itemize}

Note that the channel resource consumption argument crucially uses the behavioral interconnection view-point.

Note further, that the channel resource consumption in the system is being considered. The \emph{total} energy consumption in the architecture for reliable communication over $c$ in the above argument will, in general, be larger than the resource consumption in the original architecture because energy is consumed in the circuits in the encoder and decoder

Note that the argument holds for a channel which evolves continuously in space and time, too. This is because of the black-box nature of the proof. 

{ In fact, it is also true that for every $\epsilon > 0$, there exists a channel $c$ which directly communicates the i.i.d. $X$ source within distortion $D$ but its capacity is less than $R^I_X(D) + \epsilon$. A rough argument is the following: Consider a rate $R^I_X(D) + \epsilon$ deterministic source-code $<e^n, f^n>_1^\infty$ which codes the i.i.d. $X$ source within distortion $D$. Such a source-code exists by the rate-distortion theorem. Then, $<e^n \circ f^n>_1^\infty$ is a channel which directly communicates the i.i.d. $X$ source within distortion $D$. It can be proved that the capacity of this channel is $\leq R^I_X(D) + \epsilon$. This is because the output of $e^n$ has cardinality $2^{n(R^I_X(D) + \epsilon)}$. A full proof is omitted because it is not needed for the results proved in this paper.
}

The use of random codes is crucial to the argument. In Shannon's original argument \cite{ShannonReliable}, the existence of a random code implies the existence of a  deterministic code. However, the same is not the case here because the channel belongs to a set. In fact, if random codes are not permitted, the above theorem does not hold, see \cite{MukulSwastikMitter}.


Note that the channel model is completely general but the source model is i.i.d.The proof for general stationary ergodicity sources which satisfy a mixing requirement can be found in \cite{paperTogetherErgodic}.

\section{Universal source-channel separation theorem for communication with a fidelity criterion} \label{UUTT}

In this section, a universal source-channel separation theorem for communication with a fidelity criterion is proved. Universality is over the channel, that is channel may belong to a set; in other words, the channel may be compound. The source model is assumed to be i.i.d. 

Source-channel separation for i.i.d. sources (or in general, sources satisfying an ergodicity requirement) was proved by Shannon \cite{Shannon}. Various generalizations exist, for example, \cite{Elkayam} proves a source-channel separation for sources satisfying a sphere packing optimality condition and general channels. The contribution of this section is to prove results for compound, { general} channels and further, the results can be generalized to the multi-user setting as will be seen in the next section.

First, some notation particular to this section is defined.  This is followed by a statement and proof of the separation theorem followed by comments.

\subsection{Some notation}

Channels with input space $\mathbb I$ and output space $\mathbb O$ will be considered where $\mathbb I$ and $\mathbb O$ are finite sets. Channels with input space $\mathbb I$ and $\mathbb O$ will be denoted by $k$. The channel is a sequence $k = <k^n>_1^\infty$ where
\begin{align}
k^n &: \mathbb I^n \rightarrow \mathbb P(\mathbb O^n) \\
       &: \iota^n \rightarrow k^n(\cdot|\iota^n)
\end{align}

A channel set $k \in \mathbb B$ is said to \emph{be capable of} of communicating the i.i.d. $X$ source within distortion $D$ if there exist encoder $e = <e^n>_1^\infty$ with input space $ \mathbb X$ and output space $\mathbb I$ and a decoder with input space $\mathbb O$ and output space $\mathbb Y$, independent of the particular $k \in \mathbb A$ such that $c  \in \mathbb A \triangleq \{ <e^n \circ k^n \circ f^n>_1^\infty \ | k \in \mathbb B\}$ directly communicates the i.i.d. $X$ source with distortion $D$. The encoder and decoder are allowed to be random.

Note the difference between a channel which directly communicates the i.i.d. $X$ source within distortion $D$ and a channel which is capable of communicating the i.i.d. $X$ source within distortion $D$. 

The channel $k$ should be thought of as a `physical' channel even though the result of the next subsection will hold for general $k$.

Optimality of source-channel separation for communication with a fidelity criterion refers to the fact that given a channel $k \in \mathbb B$ which is capable of communicating the i.i.d. $X$ source within distortion $D$, the channel $k \in \mathbb B$ is also capable of communicating the i.i.d. $X$ source via a source-channel separation based architecture. That is, there exist source encoder-decoder $<e_s^n, f_s^n>_1^\infty$ and channel encoder-decoder $< e_{ch}^n, f_{ch}^n>_1^\infty$ such that $\{ e_s^n \circ  e_{ch}^n \circ k^n \circ f_{ch}^n \circ f_s^n>_1^\infty\ | \ k \in \mathbb B \}$ directly communicates the i.i.d. $X$ source within distortion $D$ and further, with the use of $< e_{ch}^n, f_{ch}^n>_1^\infty$, reliable communication happens over $k \in \mathbb B$.

The importance of separation architectures is discussed in \cite{GallagerDigital}.

\subsection{Source-channel separation}

In this subsection, a universal source-channel separation theorem for communication with a fidelity criterion is proved in the point-to-point setting, where the universality is over the channel. This is followed by comments.

\begin{thm}[Optimality of separation] \label{KKUniXKK}
Assume that random codes are permitted. Consider a set of channels $k \in \mathbb B$. Let $k \in \mathbb B$ be capable of communicating the i.i.d. $X$ source within distortion $D$. Then, reliable communication can be accomplished over $k \in \mathbb B$ at rates $<R^I_X(D)$. This reliable communication can be accomplished with consumption of channel resources which is the same as the channel resource consumption in the original architecture which communicates the i.i.d. $X$ source with distortion $D$ over  $k \in \mathbb B$. 

Further, if  reliable communication can be accomplished over $k \in \mathbb B$ at a certain rate strictly $>R^I_X(D)$, then the i.i.d. $X$ source can be communicated with  distortion $D$ over $k \in \mathbb B$ by use of a separation architecture. The channel resource consumption in this separation architecture is the same as the channel resource consumption in the architecture for reliable communication at rate strictly $>R^I_X(D)$ when the distribution on the message set is uniform.
\end{thm}
\begin{proof}
Channel set $k \in \mathbb A$ is capable of communicating the i.i.d. $X$ source within distortion $D$. This is accomplished with the help of some encoder-decoder $<e^n, f^n>_1^\infty$. Consider the channel set
\begin{align}
c \in \mathbb A \triangleq \{ <e^n \circ k^n \circ f^n>_1^\infty, k \in \mathbb A\}
\end{align}
Channel set $c \in \mathbb A$ is a compound channel which directly communicates the i.i.d. $X$ source within distortion $D$. By Theorem \ref{BasicTheorem}, it follows that reliable communication can be accomplished over $c$ at rates $<R^I_X(D)$ by use of some encoder-decoder $<E^n, F^n>_1^\infty$. It follows that reliable communication at rates $<R^I_X(D)$ can be accomplished over $k \in \mathbb A$ at rates $<R^I_X(D)$ by use of encoder-decoder $< e_{ch}^n, f_{ch}^n>_1^\infty \triangleq <E^n \circ e^n, f^n \circ F^n>_1^\infty$. The argument concerning resource consumption is the same as in the proof of Theorem \ref{BasicTheorem}. 

The proof of the second part of the theorem is the usual argument of source coding followed by channel coding. Briefly, the argument is the following: let rate $R^I_X(D) + \epsilon$ be reliably achievable over $k \in \mathbb A$. Encode the i.i.d. $X$ source within distortion $D$ at rate $R^I_X(D) + \epsilon$. This can be done because $R^P_X(D) = R^I_X(D)$. Communiate this rate $R^I_X(D) + \epsilon$ message reliably over $k \in \mathbb A$. Then, source-decode the message reproduction. End-to-end, the i.i.d. $X$ source is communicated within a distortion $D$ over $k \in \mathbb A$. A detailed argument is omitted. Since random codes are permitted, by using a symmetrically permuted codebook, the distribution on the output of the source encoder can be made uniform. It follows by an argument similar to the resource consumption argument before that the channel resource consumption in the separation architecture is the same as the channel resource consumption in the architecture for reliable communication at rate $R^I_X(D) + \epsilon$ when the distribution on the message set is uniform. A full argument is omitted.

\end{proof}

\subsection{Comments}

The argument uses a layered black-box behavioral interconnection perspective:
\begin{itemize}
\item
The perspective is black-box behavioral interconnection based for the same reason as in the previous section.
\item
The perspective is layered because $<E^n, F^n>_1^\infty$ is layered `on top of' $<e^n, f^n>_1^\infty$ in order to get the encoder-decoder $< e_{ch}^n, f_{ch}^n>_1^\infty = <E^n \circ e^n, f^n \circ F^n>_1^\infty$
\end{itemize}

Optimality of source-channel separation for universal communication with a fidelity criterion fails to hold if random codes are not permitted for the same reason as the previous section. The details are left to a future paper.

{ Note: \cite{VerduHan} provides an example of a non-ergodic channel over which source-channel separation does not hold. Roughly, the example is the following: consider a binary channel where the output codeword is equal to the transmitted codeword with probability $\frac{1}{2}$ and independent of the transmitted codeword with probability $\frac{1}{2}$. If the input to the channel is a bit source and the distortion measure is hamming, the expected distortion between input and output is approximately $\frac{1}{4}$. However, the capacity of this channel is zero and communication at distortion $\frac{1}{4}$ cannot be carried out via a source-channel separation architecture. This example seems to contradict the result in this paper. This is not the case because \cite{VerduHan} uses the expected distortion criterion whereas the probability of excess distortion criterion is used in this paper.}

The results of the previous  section and this section are generalized to the multi-user setting in the next section.

\section{Generalization to the unicast multi-user setting} \label{MMTT}

The source-channel separation theorem for communication with a fidelity criterion when the medium model is compound and general is proved. The source model is assumed to be i.i.d.

If the medium model is assumed to be memoryless, and sources are independent of each other, a source-medium separation for the rate-distortion problem was proved in \cite{Tian}. If sources are not independent of each other, it is known that source-channel separation is not optimal, see \cite{Gastpar}, Chapter 1. The contribution of this section is to prove separation when the medium model is general and compound. { This result on optimality of separation in the multi-user setting is not known, for either a general medium, or a compound medium, even separately, to the best of knowledge of the authors.}

The discussion in this section will be brief. This is because the results are a simple generalization of the point-to-point setting.

First, some notation particular to this section is defined. This is followed by a statement and proof of the separation theorem followed by comments.

\subsection{Basic theorem for the unicast multi-user setting}

\subsubsection{Basic notation for the multi-user setting}

There are $N$ users.

Let $\mathbb X_{ij}, \mathbb Y_{ij}$ $1 \leq i,j \leq N, i \neq j$ be finite sets. 
The source input spaces are $\mathbb X_{ij}$ and source reproduction spaces are $\mathbb Y_{ij}$. $\mathbb X_{ij}$ are also the channel input spaces and $\mathbb Y_{ij}$ are also the channel reproduction spaces.

In the above, $i \neq j$ and this will continue to be the case. This assumption will not be explicitly stated. Further, unless there is the possibility of confusion, $1 \leq i, j \leq N$ will not be stated either.

 Let $X_{ij}$ be a random variable of $\mathbb X_{ij}$. $X^n_{ij}$ is the i.i.d. $X_{ij}$ sequence of length $n$. $<X^n_{ij}>_1^\infty$ is the i.i.d. $X_{ij}$ source. Source $X_{ij}$ is communicated from user $i$ to user $j$. 

$d_{ij}: \mathbb X_{ij} \times \mathbb Y_{ij} \rightarrow [0, \infty )$ is a distortion function. The $n$ letter distortion function is defined additively as in Section \ref{NotandDef}. $d_{ij}$ is the distortion function for communication from user $i$ to user $j$.

The medium is a transition probability $<w^n>_1^\infty$ where
\begin{align} \label{GeneralTP}
w^n\ : \ \prod_{i, j = 1}^N \mathbb X_{ij}^n \rightarrow \mathbb P \left ( \prod_{i,j=1}^N \mathbb Y_{ij}^n \right )
\end{align}
with the meaning of this transition probability the obvious generalization from Section \ref{NotandDef}.

As in the point-to-point setting, $w = <w^n>_1^\infty \in \mathbb A$ for some set $\mathbb A$ with input spaces $\mathbb X_{ij}$ and output spaces $\mathbb Y_{ij}$.

$w \in \mathbb A$ is said to \emph{directly} communicate the i.i.d. $X_{ij}$ sources within distortions $D_{ij}$ if with inputs $X_{ij}$, $1 \leq j \leq N$ at user $i$, the output at user $j$ corresponding to source $X_{ij}$ is $Y_{ij}$ such that 
\begin{align}
\Pr \left ( \frac{1}{n} d^n_{ij}(X^n_{ij}, Y^n_{ij} > D_{ij} \right ) \leq \omega_{n, ij} \quad \forall i, j, \ \forall w \in \mathbb A
\end{align}
for some sequences $\omega_{ij, n} = <\omega_{ij, n} >_1^\infty$, $\omega{ij,n} \to 0$ ans $n \to \infty$ for all $i,j$.

Reliable communication from user $i$ to user $j$ is defined as in the point-to-point setting; messages transmitted from user $i$ to user $j$, $1 \leq i, j \leq N$ are assumed to be independent. Precise details are omitted.

\emph{In what follows, it is crucial that sources $X_{ij}$ are independent of each other. Thus, in particular, only independent sources need to be reproduced at different destinations. In particular, the same source does not need to be reproduced at different destiinations, and the setting is unicast.}

\begin{thm} \label{BasicMultiUser}
Let $w \in \mathbb A$ communicate independent  $X_{ij}$ sources within distortions $D_{ij}$ from user $i$ to user $j$. Assuming random codes are permitted, reliable communication can be accomplished from user $i$ to user $j$ at rates $<R^I_{X_{ij}}(D)$, $1 \leq i, j \leq N$. This reliable communication can be accomplished by the consumption of the same channel resources as in the original architecture for communication of i.i.d. $X_{ij}$ sources.
\end{thm} 
\begin{proof}
Consider two users, user $s$ and user $r$. Source $X_{sr}$ is communicated from user $s$ to user $r$ with distortion $D_{sr}$. By Theorem \ref{BasicTheorem}, it follows that reliable communication can be accomplished from user $s$ to user $r$ at rates $<R^I_{X_{sr}}(D_{sr})$. By the encoder construction in the proof of Theorem \ref{BasicTheorem}, the input to the medium at terminal $s$ corresponding to source $X_{sr}$ in the architecture for reliable communication is some $X'_{sr}$ where $X'_{sr}$ has the same distribution as $X_{sr}$. It follows that the joint distribution at the medium terminals is not changed (i.i.d. $X_{ij}$ where $X_{ij}$ are independent of each other). By use of an induction argument, it follows that reliable communication can be accomplished from user $i$ to user $j$ at rates $< R^I_{X_{sr}}(D_{sr})$.

The argument concerning resource consumption is the same as in the point-to-point setting.
\end{proof}

Note that the behavioral view is  used: in order to carry out the induction argument, it is necessary that the distribution at the medium input terminals does not change, and this is a consequence of the behavioral interconnection view which matches distributions at the input terminals.

Next, source-medium separation is considered.

Let $\mathbb I_i$ and $\mathbb O_i$, $1 \leq i \leq N$ be finite sets. $\mathbb I_i$ are the medium input spaces and $\mathbb O_i$ are the medium output spaces.

The medium $m = <m^n>_1^\infty$ is a transition probability for each $n$, 
\begin{align}
m^n: \prod_{i = 1}^N \mathbb I_i^n  \rightarrow \mathbb P(\prod_{i = 1}^N \mathbb O_i^n)
\end{align}

The medium belongs to some set $m \in \mathbb B$.

Communication of sources $X_{ij}$ over $m \in \mathbb B$ requires modems $h_i =\linebreak <h_i^n>_1^\infty$ at user $i$. The modems are allowed to collectively generate random codes. The exact model of the modems is irrelevant; it suffices to sa that the interconnection of medium and modems can be thought of as a transition probability of the sort (\ref{GeneralTP}).

Medium $m \in \mathbb B$ is said to \emph{be capable of}  communicating the i.i.d. $X_{ij}$ sources within distortion $D_{ij}$ if the interconnection of the medium and the modems directly communicates the i.i.d. $X_{ij}$ sources within distortion $D_{ij}$.

Next, source-medium separation is discussed. Note that the medium is allowed to be compound, in other words, there is universality over the medium. Further note that \emph{it is crucial that} that assumption $X_{ij}$ are independent of each other be made; otherwise, source-medium separation does not hold.

\begin{thm}[Optimality of separation: multi-user setting]
\label{MUUSCSMultiUserSettingMU}
Assume that random codes are permitted. Let $m \in \mathbb B$ be  capable of universally communicating i.i.d. $X_{ij}$ source from user $i$ to user $j$, $1 \leq i, j \leq N$ to within distortion levels $D_{ij}$ under additive distortion functions $d_{ij}$. The sources $X_{ij}$ are independent of each other. Then, reliable communication can be accomplished from user $i$ to user $j$, $1 \leq i, j \leq N$ over $m \in \mathbb B$ at rates $R_{ij} < R^P_{X_{ij}}(D_{ij})$. This reliable communication can be accomplished by consumption of same or lesser medium resources at each user as the medium resource consumption in the original architecture for communication of i.i.d. $X_{ij}$ sources to within distortion levels $D_{ij}$.

Further, if reliable communication can be accomplished over  $m \in \mathbb B$ from user $i$ to user $j$, $1 \leq i, j \leq N$ at certain rates strictly larger than $R^P_{X_{ij}}(D_{ij})$, then the independent, i.i.d. $X_{ij}$ sources  can be communicated from user $i$ to user $j$, $1 \leq i, j \leq N$ to within distortion levels $D_{ij}$ over $m \in \mathbb B$ by use of a separation architecture. The medium resource consumption in this separation architecture at each user is the same as or lesser than the medium resource consumption in the architecture for reliable communication at rate strictly $>R^P_{X_{ij}}(D_{ij})$ from user $i$ to user $j$ when all messages for transmission between all pairs of users $(i, j)$ are independent of each other and the distribution on every message is uniform.
\end{thm}
\begin{proof}
Follows from Theorem \ref{MUUSCSMultiUserSettingMU} in exactly the same way as Theorem \ref{KKUniXKK} follows from Theorem \ref{BasicTheorem}.
\end{proof}

\section{Other applications} \label{KKKKKKK}

Applications of a  conceptual nature are discussed in this section.

\subsection{A randomized covering-packing duality and an\\ operational view} \label{OOVV}

In \cite{paperDuality}, a duality between source and channel coding is shown between source-coding and channel-coding. This duality is a sort of a covering-packing duality, but in a randomized sense. Further, this is an operational view in the sense that the definition of channel capacity as the maximum rate of reliable communication and rate-distortion function as the minimum rate needed to code a source within a certain distortion are used. Single letter simplifications in terms of maximum and minimum mutual information respectively are not used. Please see see \cite{paperDuality} for details.

\subsection{A proof of the converse part of source-channel separation using achievability techniques}  \label{AATT}

Note Theorem \ref{KKUniXKK}. The first part of the theorem is what is referred to as the converse part of source-channel separation theorem for communication with a fidelity criterion. This is proved using a random coding argument as opposed to converse techniques of manipulations with entropies and mutual informations, see for example, in \cite{Shannon}. This proof of a theorem which is traditionally viewed as converse by use of achievability techniques is interesting in its own right.

\subsection{An equivalence perspective in point-to-point communication} \label{Equ}

Bits are used as the fundamental currency leading to a block of reliable communication in the sense that all architectures are built above an architecture for reliable communication. One of the reasons for this is that if a channel is known to transmit bits reliably, then, that channel can be used for various other purposes (example: communication within a certain distortion) without using explicit knowledge of the channel (there are other reasons, for example, two such bit pipes can be connected and the end-to-end result is also reliable communication; however, let us restrict just to this reason here). However, a pipe which communicates a source $X$ to within a certain distortion $D$ can also be used as a universal building block in the point-to-point setting assuming random codes are permitted. This is because, by Theorem \ref{BBTT}, reliable communication can be accomplished over such a `pipe' to within rates $<R^E_X(D)$ without knowledge of the channel, and this implies that if $X'$ is another source for which $R^E_{X'}(D') < R^E_X(D)$, the source $X'$ can also be communicated over this pipe by use of the usual argument of source-coding followed by channel-coding. Thus, in the point-to-point setting, the problems of reliable communication and communication within a certain distortion are equivalent in the sense that either architecture can be used as the primary building block for communication. Note that universality over the channel in Theorem \ref{BBTT} is crucial in order to prove this equivalence.

\section{Discussion and recapitulation} \label{RR}

Layered black-box behavioral interconnection is discussed, and illustrated through the problem of point-to-point and unicast multi-user communication with fidelity criteria when the channel or media may belong to a set and the set consists of general channels or media.  The problem of communication with a fidelity criteria over general (in the sense of \cite{VerduHan}), compound media are reduced to the problem of reliable communication over the media. This is a reductionist view in the sense that the problem of communication with fidelity criteria is reduced to the problem of reliable communication; no results are presented for capacity of channels or networks. 

A behavioral interconnection perspective has been useful in a control context. The authors hope that the layered black-box behavioral interconnection perspective discussed and illustrated through this paper by an example of a communication problem  will be useful in other communication and control problems.

\section*{Acknowledgements}

The authors thank Dr. Shashibhushan Borade for initial discussions which led to the proof of Theorem \ref{BasicTheorem}. The authors thank Prof. Ashish Khisti for discussions which led to the construction of a channel for which rates $>R^P_X(D) + \epsilon$ are not reliably achievable, as has been discussed in Section~4.2. The authors thanks Dr. Tom Richardson and Dr. Rajiv Laroia for hosting the first author for a semester at Qualcomm Flarion Technologies which made him think about network generalizations of Theorem \ref{BasicTheorem}, thus leading to Theorem \ref{BasicMultiUser}. The authors thanks Prof. Emeritus Robert Gallager for many insightful discussions concerning the results in this paper. Finally, the authors thank Prof. John Tsitsiklis and Prof. Vincent Tan for various discussions at various points.

\bibliography{1-Agarwal}

\address{Mumbai, Maharashtra, 410210, India\\
\email{magar@alum.mit.edu}}

\address{Laboratory for information and decision systems \\
Department of Electrical Engineering and Computer Science \\
Massachusetts Institute of Technology\\
Cambridge, MA 02139-4307, USA\\
\email{mitter@mit.edu}}

\address{Department of Electrical Engineering and Computer Sciences\\
University of California, Berkeley\\
 Berkeley, CA 94720-1770, USA\\
\email{sahai@eecs.berkeley.edu}}

\end{document}